\documentclass[aps,pra,10pt,twocolumn,superscriptaddress,notitlepage,footinbib,longbibliography]{revtex4-2}

\usepackage{natbib}
\usepackage{physics}
\usepackage[dvipsnames]{xcolor}
\usepackage{amsmath, amsfonts, bm}
\usepackage{siunitx}
\usepackage{graphicx}
\usepackage[english]{babel}
\usepackage[utf8]{inputenc}
\usepackage[autostyle]{csquotes}\MakeOuterQuote{"}
\usepackage[shortlabels]{enumitem}
\usepackage{dsfont}
\usepackage{comment}
\usepackage{verbatim}
\usepackage{xpatch}
\usepackage{tcolorbox}
\usepackage{float}
\usepackage{etoolbox}
\usepackage{amssymb}	
\usepackage{soul}
\usepackage{amsthm,mathtools}
\usepackage{postnotes}

\usepackage[hidelinks]{hyperref}
\usepackage[capitalize]{cleveref}

\crefname{appendix}{Appendix}{Appendices}

\makeatletter
\AddToHook{cmd/appendix/before}{\def\cref@section@alias{appendix}\def\cref@subsection@alias{appendix}}
\makeatother

\date{\today}

\definecolor{cadetgrey}{rgb}{0.57, 0.64, 0.69}

\newtheorem{lemma}{Lemma}

\newtheorem{corollary}{Corollary}

\theoremstyle{definition}

\newcommand{\green}[1]{{\color{OliveGreen} #1}}

\newcommand{\set}[1]{\{#1\}}

\newcommand{\identity}{\mathds{1}}

\newcommand{\rv}[1]{\bm{#1}}

\newcommand{\epsleak}{\varepsilon}

\newcommand{\pB}[1]{p_{#1}^{\rm B}}
\newcommand{\piA}[1]{p_{#1}^{\rm A}}

\newcommand{\Beve}{B}
\newcommand{\Bbob}{B}
\newcommand{\Ndet}{\rv{N}_{\rm det}}

\tcbuselibrary{theorems}

\tcbset{
  myboxcolor/.style={colframe=#1!25!black}
}

\newtcbtheorem[]{mybox}{Box}%
  {colback=black!8!white, colframe=blue!25!black,
   fontupper=\small, fonttitle=\bfseries}{box}

\definecolor{cadetgrey}{rgb}{0.57, 0.64, 0.69}

\begin{document}

\title{Phase-error estimation for quantum key distribution with leaky receivers}
\author{Álvaro Navarrete}
\affiliation{Faculty of Engineering, University of Toyama, Gofuku 3190, Toyama 930-8555, Japan}
\affiliation{Vigo Quantum Communication Center, University of Vigo, Vigo E-36310, Spain}
\affiliation{School of Telecommunication Engineering, Department of Signal Theory and Communications, University of Vigo, Vigo E-36310, Spain}
\affiliation{atlanTTic Research Center, University of Vigo, Vigo E-36310, Spain}

\author{Margarida Pereira}
\author{Guillermo Currás-Lorenzo}
\affiliation{Vigo Quantum Communication Center, University of Vigo, Vigo E-36310, Spain}
\affiliation{School of Telecommunication Engineering, Department of Signal Theory and Communications, University of Vigo, Vigo E-36310, Spain}
\affiliation{atlanTTic Research Center, University of Vigo, Vigo E-36310, Spain}

\author{Akihiro Mizutani}
\affiliation{Faculty of Engineering, University of Toyama, Gofuku 3190, Toyama 930-8555, Japan}

\author{Marcos Curty}
\affiliation{Vigo Quantum Communication Center, University of Vigo, Vigo E-36310, Spain}
\affiliation{School of Telecommunication Engineering, Department of Signal Theory and Communications, University of Vigo, Vigo E-36310, Spain}
\affiliation{atlanTTic Research Center, University of Vigo, Vigo E-36310, Spain}

\author{Kiyoshi Tamaki}
\affiliation{Faculty of Engineering, University of Toyama, Gofuku 3190, Toyama 930-8555, Japan}

\begin{abstract}
Practical quantum key distribution (QKD) receivers may leak information about their measurement outcomes and settings to the channel---e.g., through detector backflashes or back-reflected Trojan-horse light---that could compromise the protocol's security. Here we present a simple finite-key security proof based on phase-error estimation for prepare-and-measure QKD in the presence of either a priori information leakage about the basis choices and/or a posteriori information leakage about the measurement outcomes. The proof requires only a bound on the distinguishability of the side-channel states. Furthermore, the analysis is modular and compatible with existing security proofs that address detector and source imperfections, making it applicable to a wide range of practical QKD implementations.
\end{abstract}

\maketitle


\paragraph*{Introduction.}
Quantum key distribution (QKD) promises information-theoretic security~\cite{bb84,lo2014secure,pirandola2020advances,xu2020secure}, but this guarantee holds only if the real devices match the idealized models assumed in security proofs~\cite{gllp2004,marquardt2024implementation,zapatero2025}. The receiver is widely regarded as the most vulnerable component of a QKD setup~\cite{zapatero2025,marcomini2024dem}: an eavesdropper (Eve) can e.g. exploit a detection-efficiency mismatch~\cite{fung2009dem}, blind or control the detectors with bright light illumination~\cite{lydersen2010hacking}, or probe Bob's measurement unit with a Trojan-horse attack (THA)~\cite{vakhitov2001,jain2014trojan}. Although interference-based QKD~\cite{lo2012mdi,lucamarini2018overcoming} closes all detector side channels, prepare-and-measure (P\&M) schemes remain preferred for short- and mid-range links in current deployments owing to their simpler implementation. Setups of this latter type require security proofs that explicitly incorporate receiver imperfections~\cite{tupkaryPhaseError2025,wangPhaseError2025,marcomini2024dem,naharImperfectDetectors2026,currasLorenzo2025detector,navarreteNumericalSecurity2026}.

Another relevant receiver imperfection is detector leakage: after Bob's measurement, his receiver may emit signals---e.g., backflash radiation from avalanche photodiodes~\cite{kurtsiefer_breakdown_2001,pinheiro2018eavesdropping}---that carry information about his recorded bit value. Crucially, a bit value is only generated, and can thus be leaked, in those rounds in which Bob observes a detection. Therefore, the associated secret-key penalty is expected to scale with the number of detected rounds, so that the tolerable leakage is essentially independent of the channel loss; we refer to a security analysis achieving this scaling as \emph{loss tolerant}. Unfortunately, the few existing security analyses addressing this imperfection do not have this property, making them extremely pessimistic in typical scenarios where most rounds are undetected. In particular, Ref.~\cite{maroySecurityQuantum2010} establishes the asymptotic security of the BB84 protocol in this setting, but the analysis is restricted to leakage acting independently in each round and imposes a penalty per transmitted round. 
More recently, Ref.~\cite{arqandMutualInformation2024} extended the entropy accumulation framework to protocols leaking information in each round. This general framework has not yet been applied to specific QKD implementations, and a direct application would again penalize every transmitted round; whether a refined treatment could avoid this is an open question. Therefore, although this type of imperfection is inherently loss tolerant, a loss-tolerant security analysis has so far been absent.

In this work, we address this problem by presenting a simple and loss-tolerant finite-key security proof based on phase-error estimation for general P\&M schemes with---either active or passive---leaky receivers, which only requires a fidelity bound between the leakage states.
Importantly, our method is modular: it extends to non-projective qubit measurements, it can be combined with analyses accounting for other detector vulnerabilities~\cite{wangPhaseError2025} and it is compatible with security proofs that incorporate general source imperfections~\cite{curras-lorenzoSecurityFramework2025,navarreteNumericalSecurity2026}. 

Moreover, we also introduce techniques to accommodate leakage of Bob's basis choice before his measurement takes place. For instance, this may occur in active receivers through back-reflected Trojan-horse light or electromagnetic side channels. To the best of our knowledge, no previous security proof considers this scenario. Unlike post-measurement bit leakage, this imperfection is not inherently loss tolerant, and our bound for it accordingly carries a penalty per transmitted round. 

\begin{figure*}
    \centering
    \includegraphics[width=0.8\textwidth]{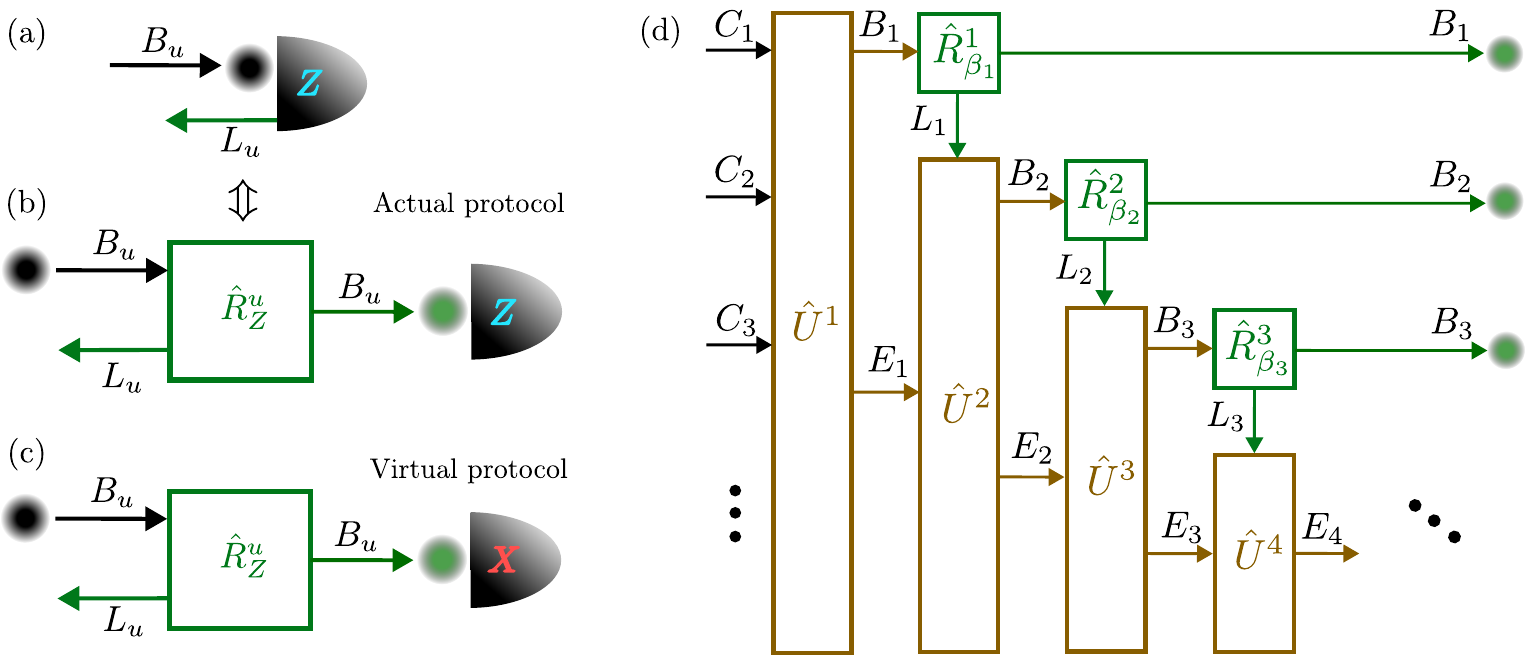}
    \caption{The isometry modeling the information leakage is designed to commute with Bob's measurement, and therefore their order can be swapped without changing the observed statistics. Thus, the actual scenario (a), in which Bob first measures $\Beve_u$ and the leakage system $L_u$ is generated afterwards, is equivalent to a fictitious scenario (b) in which the leakage isometry is applied first and Bob then measures the resulting system. In the virtual protocol (c), Bob measures the output system $\Bbob_u$ of $\hat{R}_Z^{u}$ in the $X$ basis, $\set{\ket{0_X},\ket{1_X}}$. (d) Eve's most general attack in the presence of receiver-side leakage. 
    }
    \label{fig:leakage_model}
\end{figure*}
%

\paragraph*{Protocol and assumptions.}
For concreteness, we present our analysis for the standard BB84 protocol with a single-photon source and an active receiver, although we remark the techniques we develop apply for a broader range of schemes, including decoy-state protocols~\cite{hwang2003quantum,lo2005decoy,wang2005beating}. In each round $u$, Alice selects a bit $a\in\set{0,1}$ at random and a basis $\alpha\in\set{Z,X}$ with probability $\piA{\alpha}$, and sends the corresponding BB84 state to Bob. Upon reception, Bob selects a measurement basis $\beta\in\set{Z,X}$ with probability $\pB{\beta}$ and measures the incoming system; for notational simplicity, we omit the round index $u$ in the settings and variables whenever it is clear from the context. After the quantum phase, Alice and Bob announce their basis choices: rounds in which both chose $Z$ form the sifted key, and rounds in which both chose $X$ are used as test rounds for parameter estimation.
 
Our analysis requires that the system $\Beve_u$ measured by Bob in round $u$ can be effectively described as a qutrit, consisting of a vacuum component $\ket{\perp}_{\Beve_u}$ and a two-dimensional detection subspace spanned by $\set{\ket{0}_{\Beve_u},\ket{1}_{\Beve_u}}$. Concretely, we consider that Bob's received states have a block-diagonal form satisfying
\begin{equation}
  \rho_{\Beve_u} = \mathds{1}_{\Beve_u}^{\perp}\rho_{\Beve_u}\mathds{1}_{\Beve_u}^{\perp}
             + \mathds{1}_{\Beve_u}^{\rm det}\rho_{\Beve_u}\mathds{1}_{\Beve_u}^{\rm det},
\end{equation}
where $\mathds{1}_{\Beve_u}^{\rm det}=\dyad{0}_{\Beve_u}+\dyad{1}_{\Beve_u}$ and $\mathds{1}_{\Beve_u}^{\perp}=\dyad{\perp}_{\Beve_u}$ are projectors onto the qubit (detected) and vacuum (non-detected) subspaces. For the detected rounds, we assume that Bob performs an ideal projective qubit measurement in his chosen basis $\beta$. We note, however, that this assumption can be relaxed to non-projective and imperfect qubit measurements (see \cref{app:non-projective}).
 
The qubit description for the detected states can be justified in different ways. A general route is to restrict the analysis to those rounds in which Bob receives a single photon, whose number can be estimated within the protocol~\cite{Koashi2008PracticalDetectors,moroder2009detector,Kawakami2025PassiveBiasedBB84}. This allows a direct combination of our results with security analyses that account for detector imperfections, such as mismatches in the efficiency and dark-count rates of the detectors~\cite{wangPhaseError2025} (see \cref{app:detector-imperfections}). Alternatively, for receivers that admit a squashing model~\cite{beaudrySquashingModels2008,gittsovichSquashingModel2014}---e.g., an active BB84 receiver with two threshold detectors of equal efficiency and dark-count rate---the arriving optical signals can be effectively mapped to the qutrit description above. As discussed below, this second route requires an additional assumption on the form of the leakage. Finally, although we assume ideal BB84 sources for simplicity, our technique is in principle compatible with security proofs that incorporate general source flaws, side channels and correlations, such as~\cite{curras-lorenzoSecurityFramework2025,navarreteNumericalSecurity2026}.

\paragraph*{Leakage model.}
We model receiver-side information leakage as follows. We assume that, in each round $u$, after Bob measures system $B_u$ in his selected basis $\beta$ and obtains an outcome $b'\in\set{\perp,0,1}$, his receiver emits a leakage state $\sigma_{L_u}^{(b',\beta)}$ on system $L_u$, which becomes available to Eve, as depicted in \cref{fig:leakage_model}(a). This state carries information about both $b'$ and $\beta$, and may vary arbitrarily across rounds---i.e., $\sigma_{L_u}^{(b',\beta)}$ implicitly depends on the round index $u$; in \cref{app:correlations} we show that, by making minimal modifications at the protocol level, it may even depend on the receiver's record of the previous rounds.
When the qubit description for the detected events arises from a squashing model, this leakage model additionally requires that the emissions depend on the physical detection pattern only through the pair $(b',\beta)$. Finer-grained leakage---e.g., revealing the occurrence of a double click, which the squashing map assigns to a random bit---has no counterpart in the squashed description and falls outside our analysis, which appears to pose a general problem for qubit squashing maps in the context of information leakage. No such requirement arises when the analysis is restricted to the single-photon rounds, since $b'$ is then the outcome of an actual qubit measurement.
 
Crucially, our analysis below demands only that the leakage states associated with the two $Z$-basis detection outcomes are not too distinguishable, in the sense that, for all rounds $u$,
\begin{equation}
\label{eq:leakage-assumption}
F(\sigma_{L_u}^{(0,Z)},\sigma_{L_u}^{(1,Z)}) \geq \gamma^2,
\end{equation}
for some $\gamma > 0$, where $F$ denotes the fidelity. The form of the states $\sigma_{L_u}^{(0,X)}$, $\sigma_{L_u}^{(1,X)}$ and $\sigma_{L_u}^{(\perp,\beta)}$ is irrelevant for our security proof. Intuitively, this is because the sifted key is extracted only from the $Z$-basis rounds, while the $X$-basis outcomes are announced publicly anyway and the non-detected rounds are discarded. Moreover, as shown in \cref{app:purification}, by purifying the leakage states and granting the purifying systems to Eve, we can conservatively assume, without loss of generality, that all leakage states are pure. That is, $\sigma_{L_u}^{(b,\beta)} = \ketbra{e_{b,\beta}}_{L_u}$ (with $b\in\set{0,1}$) and $\sigma_{L_u}^{(\perp,\beta)} = \ketbra{e_{\perp,\beta}}_{L_u}$, with the overlap of the $Z$-basis leakage states being \emph{exactly} $\braket{e_{1,Z}}{e_{0,Z}}_{L_u} = \gamma$, with $\gamma$ real and non-negative.

To prove security via the phase-error estimation formalism, one introduces an entanglement-based scenario---equivalent to the actual protocol---in which, in each round, Alice prepares the entangled state~\footnote{To analyze more general sources, one could instead consider the entangled state $\sum_{i}\sqrt{p_{i}}\ket{i}_{A}\ket{\varphi_i}_{C}$, where $i=0,1,\dots,n_A$ indexes the $n_A$ states $\ket{\varphi_i}_{C}$, emitted with probability $p_{i}$.} $(\ket{0_Z}_{A_u}\ket{0_Z}_{C_u}+\ket{1_Z}_{A_u}\ket{1_Z}_{C_u})/\sqrt{2}$, then she measures her local ancilla $A_u$ in basis $\alpha$ with probability $\piA{\alpha}$, and finally sends $C_u$ to Bob, who receives $B_u$ after the possible action of Eve. This equivalence allows the definition of a virtual protocol in which, in the detected $Z$ rounds, Alice and Bob instead measure their systems in the $X$ basis. The phase-error rate is defined as the rate of $X$-basis errors in these virtual measurements, and upper bounding it certifies the secrecy of the key generated in the actual protocol~\cite{koashi2009simple}. For the virtual protocol to be valid, however, the information available to Eve must be identical in the actual and virtual protocols. In the absence of leakage this holds trivially, and the test $X$-basis bit-error rate serves as a random sample of the phase-error rate. The leakage breaks this equivalence: if Bob measures the detected $Z$-round systems in the $X$ basis, he never obtains the $Z$-basis outcome $b$ and can no longer reproduce the leakage $\ket{e_{b,Z}}_{L_u}$ of the actual protocol.

The equivalence is restored by the following observation: since the state of the input system $B_u$ is independent of Bob's basis choice, the actual protocol [\cref{fig:leakage_model}(a)] is statistically indistinguishable from the fictitious scenario [\cref{fig:leakage_model}(b)], in which the leakage is generated \emph{before} Bob's measurement, via a basis-dependent isometry $\hat{R}^{u}_{\beta}$ acting on the received system as
\begin{equation}
\label{eq:leakage_unitaries}
\begin{aligned}
  &\hat{R}^{u}_{\beta}\ket{b_{\beta}}_{\Beve_u}
  =\ket{b_{\beta}}_{B_u}\ket{e_{b,\beta}}_{L_u},\\
  &\hat{R}^{u}_{\beta}\ket{\perp}_{\Beve_u}
  =\ket{\perp}_{B_u}\ket{e_{\perp,\beta}}_{L_u},
\end{aligned}
\end{equation}
for $b\in\set{0,1}$ and $\beta\in\set{Z,X}$, with $L_u$ handed to Eve. Bob then checks whether a detection occurred by using a filtering operation $\set{\mathds{1}_{\Beve_u}-\dyad{\perp}_{\Beve_u},\dyad{\perp}_{\Beve_u}}$ and, if so, performs an ideal qubit measurement in basis $\beta$ on $\Bbob_u$. Indeed, since $\hat{R}^{u}_{\beta}$ maps each measurement-basis state $\ket{b_{\beta}}$ to itself tensored with a leakage state, it commutes with Bob's projective measurement in basis $\beta$, so the statistics of Bob's outcomes are unchanged; and, conditioned on a given basis choice $\beta$ and outcome $b$, the leakage system $L_u$ received by Eve is the same in both scenarios.

\paragraph*{Effective phase-error operator.}
Building on the scenario of \cref{fig:leakage_model}(b), we define the virtual phase-error-estimation protocol shown in \cref{fig:leakage_model}(c). It differs only in the detected $Z$ rounds, where Bob's $Z$-basis measurement following the isometry $\hat{R}^{u}_{Z}$ is replaced by an $X$-basis measurement. Because the isometry is applied before the measurement, this replacement leaves $L_u$---and hence Eve's information---untouched, restoring the equivalence to the actual scenario.

This replacement, however, has one important consequence. In the actual $X$-basis test rounds, the preceding isometry is $\hat{R}^{u}_{X}$, which by \cref{eq:leakage_unitaries} commutes with the subsequent $X$-basis measurement and does not affect its outcome. In the virtual $X$-basis measurement of the key rounds, by contrast, the preceding isometry is $\hat{R}^{u}_{Z}$, which correlates the leakage system with Bob's $Z$-basis value; once $L_u$ is traced out, this acts as a dephasing channel in the $Z$ basis, which the phase-error operator of the virtual protocol must account for. Concretely, recall that the phase-error operator is given by
\begin{equation}\label{eq:phase-error-operator}
\begin{split}
    \hat{E}^{\rm ph}_{A\Bbob} & =\dyad{0_X}_{A_u}\otimes\dyad{1_X}_{\Bbob_u} \\ 
    &\;\;+ \dyad{1_X}_{A_u}\otimes\dyad{0_X}_{\Bbob_u}.
\end{split}
\end{equation}
Tracing the leakage system $L_u$ out of $\hat{R}_Z^u$ in \cref{eq:leakage_unitaries} defines an effective channel $\mathcal{E}_\gamma$ on $\Beve_u$ that multiplies the off-diagonal elements of the qubit block by $\gamma$ while leaving the diagonal elements and the vacuum component untouched: on the detection subspace, this is a dephasing channel $\mathcal{E}_\gamma(\rho)=\tfrac{1+\gamma}{2}\rho+\tfrac{1-\gamma}{2}Z\rho Z$. In particular, $\mathcal{E}_\gamma$ is self-adjoint, i.e., $\Tr[M\mathcal{E}_\gamma(\rho)]=\Tr[\mathcal{E}_\gamma(M)\rho]$ for all operators $M$ and states $\rho$. The effective phase-error operator is therefore obtained by applying $\mathcal{E}_\gamma$ to Bob's $X$-basis projectors in \cref{eq:phase-error-operator}, $\mathcal{E}_\gamma(\dyad{b_X}_{\Bbob})=\tfrac{1}{2}\big(\mathds{1}^{\rm det}_{\Beve}+(-1)^b\gamma X_{\Beve}\big)$, resulting in
\begin{equation}\label{eq:phase-error-operator-relation}
    \hat{E}^{{\rm ph},\gamma}_{A\Beve} = (\mathds{I}_A \otimes \mathcal{E}_\gamma)(\hat{E}^{\rm ph}_{A\Bbob})
    = \gamma\,\hat{E}^{\rm ph}_{A\Beve} + \frac{1-\gamma}{2}\,\mathds{1}_{A}\otimes\mathds{1}_{\Beve}^{\rm det}.
\end{equation}

\paragraph*{Phase-error bound.}
Consider that Bob applies the filtering operation $\set{\mathds{1}_{\Beve_u}-\dyad{\perp}_{\Beve_u},\dyad{\perp}_{\Beve_u}}$ to each system $\Beve_u$, which determines whether round $u$ is detected. This allows him to identify the set of detected rounds $\rv{\mathcal{N}}_{\rm det}$, with $\abs{\rv{\mathcal{N}}_{\rm det}}=\Ndet$; from now on, all random variables and events are implicitly conditioned on the observation of $\rv{\mathcal{N}}_{\rm det}$. Alice and Bob then measure the detected systems one by one, in the order of their round indices. Let the binary random variables $\rv{\chi}_{\rm ph}^{u}$ and $\rv{\chi}_{\rm x,er}^{u}$ indicate whether the $u$-th round yields a phase error or an $X$-basis error, respectively, let $\mathcal{F}_{u-1}^{\rm det}$ be a filtration containing all classical information generated before Alice and Bob measure the detected systems of the $u$-th round, and let $\rho_{A_u\Beve_u}^{\mathcal{F}_{u-1}^{\rm det}}$ be the corresponding conditional state. Since the $X$-basis bit-error operator of the test rounds coincides with $\hat{E}^{\rm ph}_{A\Beve}$ (as $\hat{R}^{u}_{X}$ commutes with the $X$-basis measurement), we have $\Pr[\rv{\chi}_{\rm x,er}^{u}|\mathcal{F}_{u-1}^{\rm det}]=\piA{X}\pB{X}\Tr[\hat{E}^{\rm ph}_{A\Beve}\rho_{A_u\Beve_u}^{\mathcal{F}_{u-1}^{\rm det}}]$, while $\Pr[\rv{\chi}_{\rm ph}^{u}|\mathcal{F}_{u-1}^{\rm det}]=\piA{Z}\pB{Z}\Tr[\hat{E}^{{\rm ph},\gamma}_{A\Beve}\rho_{A_u\Beve_u}^{\mathcal{F}_{u-1}^{\rm det}}]$. Combining these with \cref{eq:phase-error-operator-relation}, and noting that the conditional state is supported on the detection subspace, yields
\begin{equation}
\label{eq:round_u_relation_2}
\begin{split}
    \Pr[\rv{\chi}_{\rm ph}^{u}|\mathcal{F}_{u-1}^{\rm det}]
    =\,
    &
    \gamma\frac{p_Z^Ap_Z^B}{p_X^Ap_X^B}\Pr[\rv{\chi}_{\rm x,er}^{u}|\mathcal{F}_{u-1}^{\rm det}] + p_Z^Ap_Z^B\frac{1-\gamma}{2}.
\end{split}
\end{equation}
Summing both sides over the detected rounds and applying Azuma's inequality~\cite{azuma} (see \cref{corollary:azuma} in \cref{app:azuma}) to relate the sums of conditional probabilities to the corresponding sums of random variables, $\rv{N}_{\rm ph}:=\sum_{u\in\rv{\mathcal{N}}_{\rm det}}\rv{\chi}_{\rm ph}^{u}$ and $\rv{N}_{\rm x,er}:=\sum_{u\in\rv{\mathcal{N}}_{\rm det}}\rv{\chi}_{\rm x,er}^{u}$, we obtain that
\begin{equation}
\label{eq:Nph_bound}
\begin{split}
    \rv{N}_{\rm ph}
    \leq\,
    &
    \gamma\frac{p_Z^Ap_Z^B}{p_X^Ap_X^B}\rv{N}_{\rm x,er}
    + p_Z^Ap_Z^B\frac{1-\gamma}{2}\Ndet \\
    &+ \bigg(1+\gamma \frac{p_Z^Ap_Z^B}{p_X^Ap_X^B}\bigg)\Delta_{\rm A}(\Ndet,\epsilon/2)
\end{split}
\end{equation}
holds except with probability at most $\epsilon$, with $\Delta_{\rm A}(n,\epsilon):=\sqrt{2n\ln(\epsilon^{-1})}$. Alternatively, one could employ either Kato's inequality~\cite{kato} or the inequality recently introduced in~\cite{mannalath2026quantum}, which yield better finite-size performance when a priori estimates of the protocol statistics are available. Importantly, even though \cref{eq:Nph_bound} was derived conditioned on a fixed set of detected rounds $\rv{\mathcal{N}}_{\rm det}$, it holds for any realization of $\rv{\mathcal{N}}_{\rm det}$, and therefore unconditionally due to the law of total probability.

In \cref{app:four-state}, we further illustrate how this method can be adapted to different scenarios by considering a simple example of a four-state receiver under a THA~\cite{fung2009dem}.

\paragraph*{A priori leakage of the basis choice.}
We now show how to also accommodate scenarios in which Bob's receiver leaks information about his basis choice \emph{before} he receives system $\Beve_u$. In this situation the detection probability may generally depend on Bob's basis choice, so we shall index the relevant random variables by the transmitted rounds: for each $u=1,\dots,N$, the binary random variables $\rv{\chi}_{\rm ph}^{u}$, $\rv{\chi}_{\rm x,er}^{u}$ and $\rv{\chi}_{\rm sif}^{u}$ indicate whether a phase error, an $X$-basis error, or a sifted-key bit generation event occurs in the $u$-th transmitted round, and $\mathcal{F}_{u-1}$ is a filtration containing all classical information generated before that round. To model the leakage, we assume that, before Eve applies her round-$u$ interaction $\hat{U}^u$ (see \cref{fig:leakage_model}(d)), an additional system $L_{u}'$ is leaked to her in some unknown state $\sigma_{L_{u}'}^{\beta_u}$ that may depend on Bob's upcoming basis choice $\beta_u$, such that $\hat{U}^u:\mathcal{H}_{E_{u-1}}\otimes\mathcal{H}_{L_{u-1}}\otimes\mathcal{H}_{L_{u}'}\to\mathcal{H}_{E_{u}}\otimes\mathcal{H}_{B_{u}}$. The conditional state of systems $A_u\Beve_u$ then depends on $\beta_u$ as
\begin{equation}
\rho_{A_u\Beve_u}^{\beta_u,\mathcal{F}_{u-1}}
=
\Tr_{E_u}\left[\hat{U}^u \left( \rho_{A_uE_{u-1}L_{u-1}}^{\mathcal{F}_{u-1}} \otimes \sigma_{L_{u}'}^{\beta_u} \right)(\hat{U}^{u})^{\dagger}\right],
\end{equation}
for $1<u\leq N$, where $\rho_{A_uE_{u-1}L_{u-1}}^{\mathcal{F}_{u-1}}$ is the conditional state of $A_u$, Eve's system $E_{u-1}$, and the previously leaked systems $L_{u-1}$. 
That is, due to the a priori leakage system $L_u'$, the $Z$-basis and the $X$-basis states are no longer identical. Nevertheless, given that they are close, the $X$-basis error probability can still be a good estimate of the phase-error probability up to a correction term~\footnote{This type of strategy has previously been used to accommodate source imperfections~\cite{pereira2}, where a sufficiently low repetition rate had to be imposed to enforce a specific sequential structure on Eve’s attack and thereby ensure that all relevant operators appearing in the calculations were well defined~\cite{pereiraModifiedBB842023}. Importantly, no such restriction is required here, since the necessary sequential structure arises by construction (see \cref{fig:leakage_model}d), and all the operators involved are well defined.}. In particular, proceeding as before, but keeping track of the basis dependence of the state $\rho_{A_u\Beve_u}^{\beta_u,\mathcal{F}_{u-1}}$, we obtain
\begin{equation}\label{eq:relation_priori_Eph_Ex}
\begin{split}
    \Pr[\rv{\chi}_{\rm ph}^{u}|\mathcal{F}_{u-1}]
    =\,
    &
    \gamma \piA{Z}\pB{Z}\Tr[\hat{E}^{\rm ph}_{A\Beve}\rho_{A_u\Beve_u}^{Z,\mathcal F_{u-1}}]
    \\
    &
    + \frac{1-\gamma}{2} \Pr[\rv{\chi}_{\rm sif}^{u}|\mathcal{F}_{u-1}]
    \\
    \leq\,
    &
    \gamma\piA{Z}\pB{Z}G_{\delta}^{\rm U}\left(\frac{\Pr[\rv{\chi}_{\rm x,er}^{u}|\mathcal{F}_{u-1}]}{\piA{X}\pB{X}}\right)
    \\
    &
    + \frac{1-\gamma}{2} \Pr[\rv{\chi}_{\rm sif}^{u}|\mathcal{F}_{u-1}],
\end{split}
\end{equation}
where the equality follows from \cref{eq:phase-error-operator-relation}, with $\Pr[\rv{\chi}_{\rm sif}^{u}|\mathcal{F}_{u-1}] = \piA{Z}\pB{Z}\Tr[(\mathds{1}_{A}\otimes\mathds{1}^{\rm det}_{\Beve})\rho_{A_u\Beve_u}^{Z,\mathcal F_{u-1}}]$, and the inequality combines Eq.~(B16) of Ref.~\cite{curras2024security_phase} with $\Pr[\rv{\chi}_{\rm x,er}^{u}|\mathcal{F}_{u-1}]=\piA{X}\pB{X}\Tr[\hat{E}^{\rm ph}_{A\Beve}\rho_{A_u\Beve_u}^{X,\mathcal F_{u-1}}]$. Here,
\begin{align}
    G^{\rm U}_{\delta}(p)
    &=
    \begin{cases}
        \begin{aligned}
            &p+(1-\delta)(1-2p) \\
            &\quad + 2\sqrt{\delta(1-\delta)p(1-p)}
        \end{aligned}
        & p<\delta, \\
        1
        & \text{otherwise},
    \end{cases}
\end{align}
and $\delta$ is a parameter satisfying $\delta \leq F(\sigma_{L_{u}'}^Z,\sigma_{L_{u}'}^X)$ for all $u$, which by data processing lower-bounds $F(\rho_{A_u\Beve_u}^{Z,\mathcal{F}_{u-1}},\rho_{A_u\Beve_u}^{X,\mathcal{F}_{u-1}})$. Summing \cref{eq:relation_priori_Eph_Ex} over the $N$ transmitted rounds, upper-bounding the sum of the $G^{\rm U}_{\delta}$ terms by $N\,G^{\rm U}_{\delta}$ of the average via Jensen's inequality (using the concavity and monotonicity of $G^{\rm U}_{\delta}$), and applying Azuma's inequality to each of the resulting sums, we obtain
\begin{equation}\label{eq:priori_Nph_bound}
\begin{split}
    \rv{N}_{\rm ph}
    \leq\,
    &
    \gamma\,\piA{Z}\pB{Z}\,N\,
    G^{\rm U}_{\delta}\!\left(\frac{\rv{N}_{\rm x,er}+\Delta_{\rm A}(N,\epsilon/3)}{N\,\piA{X}\pB{X}}\right)
    \\
    &
    + \frac{1-\gamma}{2}\big(\rv{N}_{\rm sif}+\Delta_{\rm A}(N,\epsilon/3)\big)
    + \Delta_{\rm A}(N,\epsilon/3),
\end{split}
\end{equation}
which holds except with probability at most $\epsilon$, where $\rv{N}_{\rm sif}$ is the number of sifted-key bits. As before, alternative inequalities~\cite{kato,mannalath2026quantum} may replace Azuma's to improve the finite-size performance.

\paragraph*{Results.}
\begin{figure}
    \includegraphics[width=\columnwidth]{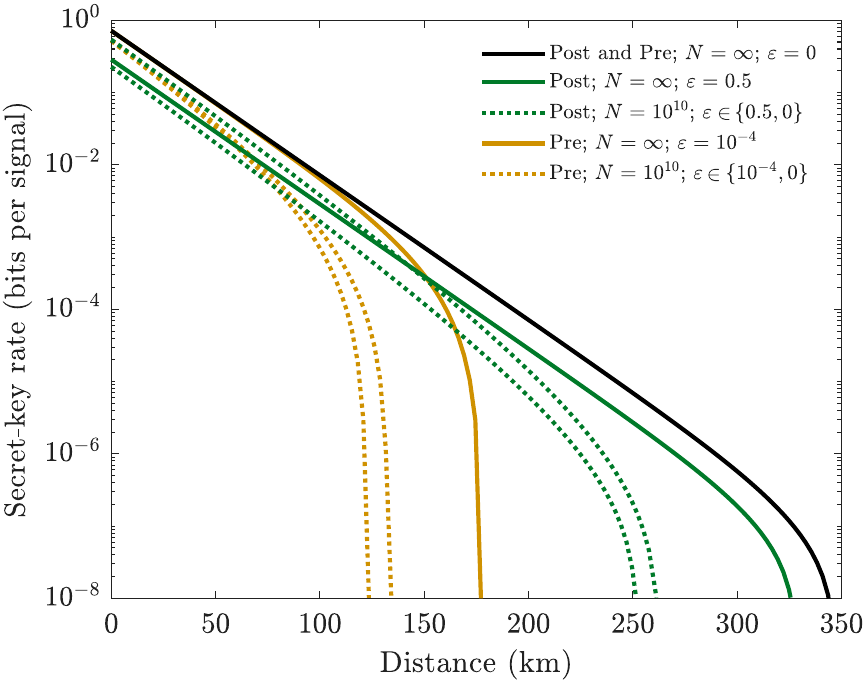}
    \caption{Secret-key rate per signal as a function of the distance for the two leakage scenarios: post-measurement leakage (Post) and a priori leakage of the basis choice (Pre). We consider a total number of transmitted signals $N=10^{10}$ and $N\to\infty$, and vary the leakage strength via the parameter $\epsleak\in\{0,\,10^{-4},\,0.5\}$, with $\epsleak:=1-\gamma^2$ (Post) and $\epsleak:=1-\delta$ (Pre). 
    }
    \label{fig:performance}
\end{figure}

To evaluate the method, we compute the secret-key rate of a single-photon BB84 protocol by considering a typical channel model with loss coefficient $0.2$~dB/km, together with a detection efficiency $\eta_{\rm d}=0.73$ and a dark-count probability $p_{\rm d}=10^{-8}$ per detector \cite{pittaluga2021600}. We optimize the $Z$-basis probability $\piA{Z}=\pB{Z}$ for each distance, and we fix the secrecy and correctness parameters to $\epsilon_{\rm sec}=\epsilon_{\rm cor}=10^{-10}$. The secret key length is computed via $\ell = \rv{N}_{\rm sif}\,[1-h(\rv{N}_{\rm ph}/\rv{N}_{\rm sif})]-\lambda_{\rm EC}-\log_2(2/\epsilon_{\rm cor})-2\log_2\!\big(1/(2\epsilon_{\rm PA})\big)$ \cite{tupkaryPhaseError2025}, where $h(\cdot)$ is the binary entropy function; $\lambda_{\rm EC}=f_{\rm EC}\,\rv{N}_{\rm sif}\,h(\mathrm{QBER})$ is the information revealed during the error correction phase, with $f_{\rm EC}=1.16$ and QBER being the quantum bit-error rate; $\epsilon_{\rm PA}=\epsilon_{\rm sec}/2$, and $\epsilon=(\epsilon_{\rm sec}/4)^2$.

The results are shown in \cref{fig:performance}. We parametrize the two leakage scenarios by a single parameter $\epsleak$, such that $\epsleak=1-\gamma^2$ for post-measurement leakage and $\epsleak=1-\delta$ for a priori leakage \footnote{Note that \cref{eq:priori_Nph_bound} accounts for both leakage mechanisms simultaneously, and therefore depends on $\gamma$ and $\delta$. For the "Pre" curves in \cref{fig:performance}, however, whether one sets $\epsleak=1-\gamma^{2}=1-\delta$---i.e. one considers that both types of leakage are present with the same strength---or removes the post-measurement leakage by setting $\gamma=1$ provides essentially the same results, since a priori leakage dominates the key-rate penalty and both choices yield essentially indistinguishable curves.}. As expected, the post-measurement leakage of Bob's outcomes affects the key rate only weakly, with a noticeable drop appearing solely when the leakage states are strongly distinguishable. A priori leakage of Bob's basis choices is considerably more harmful, as Eve can tailor her attack to exploit this information. For instance, she could perform an unambiguous state discrimination (USD) measurements on the leakage systems, and postselect only those rounds in which the USD measurement succeeds.

\paragraph*{Conclusion.}
We have presented a finite-key security proof based on phase-error estimation for P\&M QKD with leaky receivers that requires only fidelity bounds on the leakage states. Post-measurement leakage of Bob's outcomes increases the phase-error rate only by an additive loss-tolerant penalty, while pre-measurement leakage of his basis choices is considerably more harmful, as it allows Eve to fine-tune her attack strategy. Importantly, the analysis is modular and can be combined with security proofs accounting for both transmitter and detector imperfections.

\paragraph*{Acknowledgments.}
A.N. acknowledges financial support from the Xunta de Galicia (Consellería de Educación, Ciencia, Universidades e Formación Profesional) through a Xunta de Galicia Postdoctoral Fellowship (No. ED481B-2025/113). 
G.C.-L. acknowledges funding from the European Union's Horizon Europe research and innovation programme under the Marie Skłodowska-Curie Postdoctoral Fellowship grant agreement No.\ 101149523. 
A.M. was partially supported by JSPS KAKENHI Grant No. JP24K16977.
M.C. and M.P. acknowledge support from the Galician Regional Government (consolidation of Research Units: AtlantTIC); the Spanish Ministry of Science, Innovation and Universities (MICIU); the Fondo Europeo de Desarrollo Regional (FEDER) through the grant No. PID2024-162270OB-I00; the “Hub Nacional de Excelencia en Comunicaciones Cuánticas” funded by the Ministerio para la Transformación Digital y de la Función Pública and the European Union NextGenerationEU; the European Union’s Horizon Europe Framework Programme under the Marie Sklodowska-Curie Grant No. 101072637 (Project QSI); the project “Quantum Secure Networks Partnership” (QSNP, grant agreement No 101114043); the European Union under the Project IberianQCI (grant 101249593), and the Programa de Cooperación Interreg VI-A España-Portugal (POCTEP) 2021-2027 through the project QUANTUM IBERIA.
K.T. acknowledges support from JSPS KAKENHI Grant Numbers 23K25793 and 23H01096.

\appendix

\section{Extension to non-projective qubit measurements}
\label{app:non-projective}

The main text assumes that Bob performs an ideal projective qubit measurement. Here, we extend the analysis to imperfect non-projective qubit measurements. For simplicity, we omit the round index $u$ throughout this appendix. Concretely, we assume that, in each detected key round, Bob performs a (possibly unknown) two-outcome positive operator-valued measure (POVM) $\set{G^{(Z)}_0,G^{(Z)}_1}$ on a qubit and, upon obtaining outcome $b\in\set{0,1}$, his device leaks the state $\ket{e_{b,Z}}_L$ with $\braket{e_{1,Z}}{e_{0,Z}} = \gamma$. Similarly, in each detected test round, Bob performs a two-outcome POVM $\set{G^{(X)}_0,G^{(X)}_1}$, upon which his device may also leak an outcome-dependent state $\ket{e_{b,X}}_L$. However, as argued in the main text, the form of the states leaked in the test rounds is irrelevant for the security analysis---their outcomes are announced publicly anyway---so this leakage is disregarded in what follows. Like in the main text, without loss of generality, we shall assume that $\gamma$ is real and non-negative (see \cref{app:purification}).

The analysis has three main steps: (1) we show that the POVM can be replaced by a projective measurement plus classical randomness; (2) we show that the leakage of the bit value can be replaced by leakage of the outcome of the projective measurement and the classical randomness, which gives Eve at least as much information as in the actual scenario; and (3) we define a controlled isometry that implements the leakage of the outcome of the projective measurement, and use this to obtain a phase-error relation, in a similar way as in the main text.

\subsection{Step 1: Simulating the POVM by a projective measurement plus randomness}

Since $G^{(Z)}_0$ is Hermitian and acts on a qubit, it admits the decomposition
\begin{equation}
    G^{(Z)}_0=\lambda_0 P_0+\lambda_1 P_1,\qquad P_0+P_1=\mathds{1},
\end{equation}
with eigenvalues $\lambda_j\in[0,1]$ and orthogonal rank-1 projectors $\set{P_0,P_1}$; consequently $G^{(Z)}_1=\mathds{1}-G^{(Z)}_0$ is diagonal in the same basis. Bob can therefore reproduce the outcome $b$ of his POVM $\{G^{(Z)}_0,G^{(Z)}_1\}$ by measuring $\set{P_j}_j$, learning $j$, and assigning a bit $b$ by classical post-processing. We implement this post-processing with an auxiliary random variable $R\in\set{I,F,C_0,C_1}$, drawn independently of everything else with probabilities
\begin{equation}
\begin{gathered}
    p_I=\max(\lambda_0-\lambda_1,0), \quad p_F=\max(\lambda_1-\lambda_0,0),\\
    p_{C_0}=\min(\lambda_0,\lambda_1), \quad p_{C_1}=1-\max(\lambda_0,\lambda_1),
\end{gathered}
\end{equation}
which are non-negative and sum to one. Bob then sets $b$ according to~\cref{tab:table1}.
\begin{table}
\centering
\begin{tabular}{c|cc}
$R$ & $j=0$ & $j=1$ \\
\hline
$I$   & $b=0$ & $b=1$ \\
$F$   & $b=1$ & $b=0$ \\
$C_0$ & $b=0$ & $b=0$ \\
$C_1$ & $b=1$ & $b=1$ \\
\end{tabular}
\caption{Bob's bit selection according to the value of $j$ and $R$.}
\label{tab:table1}
\end{table}
Note that
\begin{equation}
\begin{gathered}
\Pr(b=0 \vert j = 0) = p_I + p_{C_0} = \lambda_0, \\
\Pr(b=0 \vert j = 1) = p_F + p_{C_0} = \lambda_1.
\end{gathered}
\end{equation}
Hence, for any input state $\rho$,
\begin{equation}
    \Pr(b=0\vert\rho)=\sum_j\lambda_j\Tr(P_j\rho)=\Tr(G^{(Z)}_0\rho),
\end{equation}
and $\Pr(b=1|\rho)=1-\Tr(G^{(Z)}_0\rho)=\Tr(G^{(Z)}_1\rho)$, so the statistics of $b$ are identical to those of the original POVM.

\subsection{Step 2: Simulating the leakage}

We now suppose that Bob runs the alternative scenario where he measures the projectors $\set{P_j}_j$ and leaks, instead of $\ket{e_{b,Z}}_L$, both the (classical) value of $R$ and a state $\ket{f_j}_L$ that carries information about the projective outcome $j$, chosen such that
\begin{equation}\label{eq:gram_match}
    \braket{f_1}{f_0}=\braket{e_{1,Z}}{e_{0,Z}}=\gamma.
\end{equation}
Because $\gamma$ is real, there exist isometries $V_I,V_F$ such that
\begin{equation}
    V_I\ket{f_j}=\ket{e_{j,Z}},\qquad V_F\ket{f_j}=\ket{e_{1-j,Z}}.
\end{equation}
Given $R$ and $\ket{f_j}$, Eve can therefore apply the following map controlled by $R$:
\begin{equation}
\begin{aligned}
    R&=I:\ \ket{f_j}\mapsto\ket{e_{j,Z}},      &\qquad R&=F:\ \ket{f_j}\mapsto\ket{e_{1-j,Z}},\\
    R&=C_0:\ \ket{f_j}\mapsto\ket{e_{0,Z}},     &\qquad R&=C_1:\ \ket{f_j}\mapsto\ket{e_{1,Z}},
\end{aligned}
\end{equation}
where for $R\in\set{C_0,C_1}$ she simply discards $\ket{f_j}$ and prepares the indicated state. Comparing with the relabeling table, in every case the output is exactly $\ket{e_{b,Z}}$, and thus this map is such that
\begin{equation}
    (R,\ket{f_j})\longmapsto\ket{e_{b,Z}}.
\end{equation}
Thus Eve can recover the actual leakage $\ket{e_{b,Z}}$ from $(R,\ket{f_j})$, so the alternative scenario gives her at least as much power as the actual one. Thus, proving the security of the alternative scenario suffices to guarantee the security of the actual scenario.

\subsection{Step 3: Isometry and dephasing channel}

Now, similarly to the main text, consider that Bob implements the following scenario, which is equivalent to the scenario just described in Step 2: (1) he applies the isometry
\begin{equation}\label{eq:iso_V}
    \hat{R}_Z\ket{\psi}_{\Beve}=P_0\ket{\psi}_{\Bbob}\ket{f_0}_L+P_1\ket{\psi}_{\Bbob}\ket{f_1}_L,
\end{equation}
which is the analogue of the isometry $\hat{R}^{u}_{Z}$ in \cref{eq:leakage_unitaries}, restricted to the detection subspace, and sends $L$ to Eve; (2) he draws $R$ and sends it to Eve; (3) he measures $\set{P_j}_j$ on system $\Bbob$, learning $j$; (4) he computes $b$ from $(j,R)$. Note that $P_iP_j=\delta_{ij}P_i$ and thus $\hat{R}_Z^\dagger \hat{R}_Z=P_0+P_1=\mathds{1}$, so $\hat{R}_Z$ is indeed an isometry.

Crucially, steps (1)-(2), which fix everything Eve receives, precede the measurement in step (3). Therefore, in the phase-error estimation protocol, we can keep steps (1)-(2) and replace steps (3)-(4) by the two-outcome POVM $\set{G^{(X)}_0,G^{(X)}_1}$ that Bob performs in the test rounds.

Similarly to the main text, we define the effective channel acting on $\Beve$ by tracing $L$ out of \cref{eq:iso_V}. We obtain
\begin{equation}
\begin{split}
   \mathcal{E}_\gamma(\rho)=
   &\ \ev*{\rho}{\tilde 0}\ketbra*{\tilde 0}
   + \ev*{\rho}{\tilde 1}\ketbra*{\tilde 1}\\
   &+ \gamma\,\mel*{\tilde 0}{\rho}{\tilde 1}\ketbra*{\tilde 0}{\tilde 1}
   + \gamma\,\mel*{\tilde 1}{\rho}{\tilde 0}\ketbra*{\tilde 1}{\tilde 0},
\end{split}
\end{equation}
with $P_0=\ketbra*{\tilde 0}$ and $P_1=\ketbra*{\tilde 1}$. This is precisely the dephasing channel $\mathcal{E}_\gamma$ of the main text but now for the $\set{\ket*{\tilde 0},\ket*{\tilde 1}}$ basis. Equivalently, we have that $\mathcal{E}_\gamma(\rho)=\tfrac{1+\gamma}{2}\rho+\tfrac{1-\gamma}{2}\tilde Z\rho \tilde Z$ with $\tilde Z=P_0-P_1$, which implies that $\mathcal{E}_\gamma$ is self-adjoint, i.e., $\mathcal{E}_\gamma^\dagger=\mathcal{E}_\gamma$. Thus,
\begin{equation}
    \Tr[\mathcal{E}_\gamma(\rho)M]=\Tr[\rho\,\mathcal{E}_\gamma(M)]
\end{equation}
for all POVM operators $M$ and states $\rho$.

\subsection{Phase-error relation}

Since $M-\mathcal{E}_\gamma(M)=(1-\gamma)\,M_{\rm off}$, where $M_{\rm off}$ is the off-diagonal part of $M$ in the $\set{\ket*{\tilde 0},\ket*{\tilde 1}}$ basis, and $\norm{M_{\rm off}}_\infty=\abs{M_{\tilde 0\tilde 1}}\le\tfrac12$ for any qubit POVM $M$, with $M_{\tilde 0\tilde 1}:=\bra*{\tilde 0}M\ket*{\tilde 1}$, we obtain the state-independent bound
\begin{equation}\label{eq:norm_bound}
    \norm{M-\mathcal{E}_\gamma(M)}_\infty\le\frac{1-\gamma}{2}.
\end{equation}
In this case, the unmodified phase-error operator has the form
\begin{equation}
\begin{split}
    \hat{E}^{\rm ph}_{A\Bbob} & =\dyad{0_X}_{A}\otimes G^{(X)}_1+ \dyad{1_X}_{A}\otimes G^{(X)}_0.
\end{split}
\end{equation}
Since the leakage in the test rounds can be disregarded, this operator coincides with the $X$-basis bit-error operator of the test rounds. In the presence of leakage in the key rounds, the modified phase-error operator has the form
\begin{equation}
\begin{gathered}
    \hat{E}^{{\rm ph},\gamma}_{A\Beve} = (\mathds{I}_A \otimes \mathcal{E}_\gamma) \big(\hat{E}^{\rm ph}_{A\Bbob}\big) \\
    = \dyad{0_X}_{A}\otimes \mathcal{E}_\gamma(G^{(X)}_1)+ \dyad{1_X}_{A}\otimes \mathcal{E}_\gamma(G^{(X)}_0).
\end{gathered}
\end{equation}
From \cref{eq:norm_bound}, it is easy to show that
\begin{equation}\label{eq:norm_bound_2}
    \norm{\hat{E}^{{\rm ph},\gamma}_{A\Beve} - \hat{E}^{\rm ph}_{A\Beve}}_\infty\le\frac{1-\gamma}{2}.
\end{equation}
We can directly use this to obtain a similar relation as \cref{eq:round_u_relation_2}, namely
\begin{equation}
\begin{split}
    \Pr[\rv{\chi}_{\rm ph}^{u}|\mathcal{F}_{u-1}^{\rm det}]
    \leq\,
    &
    \frac{p_Z^Ap_Z^B}{p_X^Ap_X^B}\Pr[\rv{\chi}_{\rm x,er}^{u}|\mathcal{F}_{u-1}^{\rm det}]
    \\
    &
    + p_Z^Ap_Z^B\frac{1-\gamma}{2}.
\end{split}
\end{equation}

\section{Combination with detector imperfections}
\label{app:detector-imperfections}
Both the analysis in the main text and in \cref{app:non-projective} require that Bob generates his sifted key bits by measuring qubits. This holds when the signals arriving at Bob are single photons, but in general Eve may resend him any number of photons she wishes, and Bob's POVM elements necessarily act on the full Fock space. There are two regimes in which the qubit assumption can nonetheless be recovered.

First, both an active and a symmetric passive BB84 receiver with ideal detectors admit a qubit squashing model~\cite{beaudrySquashingModels2008,gittsovichSquashingModel2014}. This means that Bob's measurement is equivalent to a map that squashes the incoming state to a qubit, followed by a qubit measurement. The same holds when the detectors are not ideal but share a common efficiency and dark count rate, since these common imperfections can be absorbed into the channel~\cite{naharImperfectDetectors2026}. In all these cases our results apply directly after considering the squashing model, provided that---as discussed in the main text---the leakage depends on the physical detection pattern only through the recorded outcome $b'$ and the basis $\beta$, and not on finer-grained information such as the occurrence of a double click.

Second, when the detector efficiencies and dark count rates differ, the qubit squashing model no longer applies. In this case, one can still apply our results by restricting the analysis to the rounds in which Bob receives a single photon and bounding the phase-error rate of those rounds. This is precisely the strategy of Ref.~\cite{wangPhaseError2025}, which proves security of BB84 with an asymmetric passive receiver and imperfect detectors. Below, we highlight the generality and wide applicability of our leakage analysis by showing that it can be plugged directly into the security analysis of Ref.~\cite{wangPhaseError2025}, which differs significantly from the standard BB84 analysis with ideal detectors. For concreteness we shall consider imperfect but memoryless detectors, i.e.\ Bob's POVMs are assumed to be independent across rounds~\footnote{We remark that Ref.~\cite{wangPhaseError2025} also extends its analysis to detectors with memory effects, by discarding detections in rounds preceded by another detection within the correlation length and reducing the security analysis of the non-detected rounds to the memoryless case. Since that reduction employs the same basis-efficiency-mismatch parameter $\delta$, our substitution $\delta\mapsto\delta_\gamma$ should apply unchanged at the operator level. However, the correlated-detector proof of Ref.~\cite{wangPhaseError2025} assumes that only click/no-click information is generated (and announced) during the quantum communication phase. In contrast, in the post-measurement leakage scenario considered here, some information about Bob's outcomes and basis choices becomes available to Eve during the quantum communication phase itself. While we expect this issue to be addressable by making targeted modifications to the analysis in Ref.~\cite{wangPhaseError2025}, a rigorous combination of our detector leakage analysis with the correlated-detector analysis of Ref.~\cite{wangPhaseError2025} is outside the scope of this work.}.

In the construction of Ref.~\cite{wangPhaseError2025} (see Fig.~4 in~\cite{wangPhaseError2025}), once one conditions on the rounds in which Bob receives a single photon and registers a single click in the $Z$ branch, his final key-generating measurement for those rounds is a two-outcome POVM $\set{G^{(Z)}_0,G^{(Z)}_1}$ producing the bit $b\in\set{0,1}$. This POVM is in general non-projective, but acts on a qubit since Bob received a single photon. Now suppose Bob's setup additionally leaks a state $\ket{e_{b,Z}}_L$ carrying information about $b$, with $\braket{e_{1,Z}}{e_{0,Z}}=\gamma$. This is exactly the situation analysed in \cref{app:non-projective}: the leakage is equivalent to applying a dephasing channel $\mathcal{E}_\gamma$ to Bob's qubit before his final measurement. Consequently, in the phase-error estimation protocol, the virtual $X$-basis measurement $\set{\tilde G^{(1)}_{(X,=)},\tilde G^{(1)}_{(X,\neq)}}$ used in~\cite{wangPhaseError2025} to check for phase errors in the key rounds must be replaced by its dephased version $\set{\mathcal{E}_\gamma(\tilde G^{(1)}_{(X,=)}),\mathcal{E}_\gamma(\tilde G^{(1)}_{(X,\neq)})}$, where $\mathcal{E}_\gamma$ acts on Bob's subsystem (i.e.\ $\mathcal{E}_\gamma \equiv \mathds{1}_A\otimes\mathcal{E}_\gamma$).

Thus, to incorporate detector leakage into the analysis of Ref.~\cite{wangPhaseError2025} one simply needs to substitute
\begin{equation}
\label{eq:map_leakage}
  \tilde G_{(X,\neq)}^{(1)} \;\longmapsto\; \mathcal{E}_\gamma\big(\tilde G_{(X,\neq)}^{(1)}\big),
\end{equation}
wherever $\tilde G_{(X,\neq)}^{(1)}$ refers to the \emph{phase-error} measurement on the key rounds (the bit-error measurement on the test rounds is left unchanged). In particular, the parameter controlling the phase-error rate in Ref.~\cite{wangPhaseError2025},
\begin{equation}
\delta := \left\|
\sqrt{\tilde{F}_{X}^{(1)}}\,\tilde{G}_{(X,\neq)}^{(1)}\,
\sqrt{\tilde{F}_{X}^{(1)}}
-
a\,\sqrt{\tilde{F}_{Z}^{(1)}}\,\tilde{G}_{(X,\neq)}^{(1)}\,
\sqrt{\tilde{F}_{Z}^{(1)}}
\right\|_{\infty},
\end{equation}
becomes, under the substitution in \cref{eq:map_leakage},
\begin{equation}
\begin{gathered}
\delta_\gamma := \bigg\|
\sqrt{\tilde{F}_{X}^{(1)}}\,\tilde{G}_{(X,\neq)}^{(1)}\,
\sqrt{\tilde{F}_{X}^{(1)}}  \qquad \qquad
\\
\qquad \qquad -
a\,\sqrt{\tilde{F}_{Z}^{(1)}}\,\mathcal{E}_\gamma\big(\tilde{G}_{(X,\neq)}^{(1)}\big)\,
\sqrt{\tilde{F}_{Z}^{(1)}}
\bigg\|_{\infty}.
\end{gathered}
\end{equation}
As proven below, these two quantities are related by
\begin{equation}\label{eq:delta_gamma_bound}
    \delta_\gamma \;\leq\; \delta + a\,\frac{1-\gamma}{2}.
\end{equation}
Thus, one can incorporate leakage simply by substituting $\delta\mapsto\delta_\gamma$
in the phase-error rate bound of Ref.~\cite{wangPhaseError2025}.

\begin{proof}[Proof of \cref{eq:delta_gamma_bound}]
Let
$T := \sqrt{\tilde F^{(1)}_X}\,\tilde G^{(1)}_{(X,\neq)}\,\sqrt{\tilde F^{(1)}_X}$ and
$P := \sqrt{\tilde F^{(1)}_Z}\,\tilde G^{(1)}_{(X,\neq)}\,\sqrt{\tilde F^{(1)}_Z}$,
such that $\delta=\norm{T-aP}_\infty$. Setting
$\Delta := \tilde G^{(1)}_{(X,\neq)}-\mathcal{E}_\gamma(\tilde G^{(1)}_{(X,\neq)})$
and using the triangle inequality,
\begin{equation}
\begin{gathered}
    \delta_\gamma
    = \norm{(T-aP) + a\,\sqrt{\tilde F^{(1)}_Z}\,\Delta\,\sqrt{\tilde F^{(1)}_Z}}_\infty \\
    \leq \delta + a\,\norm{\sqrt{\tilde F^{(1)}_Z}\,\Delta\,\sqrt{\tilde F^{(1)}_Z}}_\infty.
\end{gathered}
\end{equation}
Since $\tilde G^{(1)}_{(X,\neq)}$ is a POVM element, $0\le \tilde G^{(1)}_{(X,\neq)}\le\mathds{1}$, and from
\cref{eq:norm_bound} it follows that  $\norm{\Delta}_\infty\le\tfrac{1-\gamma}{2}$. Thus, $-\tfrac{1-\gamma}{2}\mathds{1}\le\Delta\le\tfrac{1-\gamma}{2}\mathds{1}$, and
conjugating by $\sqrt{\tilde F^{(1)}_Z}$ (which preserves operator ordering) gives
\begin{equation*}
    \norm{\sqrt{\tilde F^{(1)}_Z}\,\Delta\,\sqrt{\tilde F^{(1)}_Z}}_\infty
    \leq \frac{1-\gamma}{2}\,\norm{\tilde F^{(1)}_Z}_\infty
    \leq \frac{1-\gamma}{2},
\end{equation*}
where the last inequality uses $\tilde F^{(1)}_Z\le\mathds{1}$, as
$\set{\tilde F^{(1)}_X,\tilde F^{(1)}_Z}$ is a POVM. Combining the two equations above,
we obtain  \cref{eq:delta_gamma_bound}.
\end{proof}
%

\section{Correlations in the emitted leakage}
\label{app:correlations}
In the main text, the leakage emitted in round $u$ depends only on the current outcome $b'\in\set{\perp,0,1}$, the basis $\beta$, and the round index. In practice, a receiver may also emit signals in the rounds that follow a detection, so that the bit recorded in round $u$ keeps influencing the emissions of subsequent rounds. Suppose, for simplicitly, that these correlations have a finite range $L_c$, i.e., the emitted leakage state at round $u$ may depend on the outcomes and settings of rounds $u-L_c,\dots,u$. This scenario can be straightforwardly accommodated at the protocol level by requiring that Bob gates his detectors off during the $L_c$ rounds following each detection. Then no bit value is generated within the correlation range of a detected round, and the whole emission triggered by round $u$---i.e., the state emitted in that round together with those emitted in the $L_c$ subsequent, inactive rounds---depends on the recorded bit $b$ but not on prior or later outcomes. Collecting these emissions into a single system $\tilde{L}_u := L_uL_{u+1}\cdots L_{u+L_c}$, the analysis of the main text applies directly: the leakage is again described by the isometry of \cref{eq:leakage_unitaries}, now acting on $\tilde{L}_u$, and \cref{eq:Nph_bound} holds with the assumption of \cref{eq:leakage-assumption} imposed on the composite states, $F(\sigma^{(0,Z)}_{\tilde{L}_u},\sigma^{(1,Z)}_{\tilde{L}_u})\geq\gamma^2$.

The price of this countermeasure is twofold, and in both cases mild. First, the certifiable overlap may degrade, since $\gamma$ now bounds the distinguishability of the entire emission sequence rather than that of a single round. In practice, however, one expects the bulk of the emission to accompany the detection event itself, so that $\gamma$ is dominated by the overlap of the states emitted in the detected round. Moreover, the gating process may suppress part of the memory effects. Note that the correlation length characterized on a free-running receiver also accounts for correlations mediated by subsequent detection events, which would not occur while the detectors are inactive. Using such a value of $L_c$ as the gating window is therefore conservative. Second, gating discards rounds: if $p_{\rm det}$ denotes the detection probability of an active round, a fraction $\approx 1/(1+L_c\,p_{\rm det})$ of the rounds remain active. Crucially, this fraction approaches one in practical regimes, where $p_{\rm det}\ll 1/L_c$, so the sifted-key length is essentially unaffected and the loss tolerance of the analysis is preserved. A similar strategy has been employed in Ref.~\cite{wangPhaseError2025} to handle memory effects in the detectors, there by discarding in post-processing the detections occurring within the correlation length of a previous one. 

In fact, we note that our countermeasure can also be implemented in post-processing, for receivers that do not support gating. For this, in the actual protocol Alice and Bob discard every detection occurring within $L_c$ rounds of a previously kept one, so that the kept detections are again separated by more than the correlation length. In the virtual protocol, the discarded rounds are measured exactly as in the actual protocol, and only the kept ones are virtualized. The isometry associated with a kept round $u$ must then account for the emissions of the whole window, and is now controlled by $b$ jointly with the outcomes and settings of the surrounding rounds, on which the emissions may also depend~\footnote{As in the main text, this isometry is not an actual step of the protocol, but a rewriting of the window dynamics in the controlled form $\hat{S}^{u}=\sum_{b}\dyad{b_Z}_{\Beve_u}\otimes\hat{K}_b$, where $\hat{K}_b$ collects everything that happens in rounds $u,\dots,u+L_c$ once the value $b$ is recorded. This includes Bob's measurements in the discarded rounds, the emissions they trigger, and Eve's interactions in between. Such a rewriting is legitimate because $\hat{S}^{u}$ commutes with Bob's $Z$-basis measurement of that round and the reordering argument of the main text applies. One might object that, strictly speaking, the emitted leakage systems may not survive, since Eve's operations take them as inputs, and by the end of the window they have been merged into her side information (see \cref{fig:leakage_model}). However, this is not a problem. The system $\tilde{L}_u$ simply denotes the emissions as they leave the receiver, which is a well-defined object and, in any case, the natural target of a device characterization. What Eve subsequently does with them cannot increase their distinguishability, so bounding the fidelity at the point of emission is enough.}. Denoting by $h_u$ any record of the rounds $u-L_c,\dots,u-1,u+1,\dots,u+L_c$, the analysis of the main text is then recovered provided that

\begin{equation}
\label{eq:leakage-assumption-window}
F\big(\sigma^{(0,Z,h_u)}_{\tilde{L}_u},\sigma^{(1,Z,h_u)}_{\tilde{L}_u}\big)\geq\gamma^2
\end{equation}
holds for every kept round $u$ and every record $h_u$. Note that the reduction of \cref{app:purification} is now applied separately to each pair $(u,h_u)$ so that the overlap equals $\gamma$ on every branch. 
The resulting bound is provided by \cref{eq:Nph_bound} with the round counts---$\Ndet$, $\rv{N}_{\rm x,er}$---referring to the kept detections, so the performance essentially matches that of the gated implementation. The price is a more demanding device characterization, as \cref{eq:leakage-assumption-window} must be certified over records that do contain intermediate detections.

\section{Reduction to pure leakage states with a fixed overlap}
\label{app:purification}

Here we justify the claim in the main text that, without loss of generality, the leakage states can be taken to be pure, with the two $Z$-basis states having overlap exactly $\gamma$, real and non-negative. The construction below is performed separately for each round $u$.

The form of the states $\sigma_{L_u}^{(0,X)}$, $\sigma_{L_u}^{(1,X)}$ and $\sigma_{L_u}^{(\perp,\beta)}$ does not affect the security proof, so we can take them to be pure without loss of generality. As for the $Z$-basis leakage states $\sigma_{L_u}^{(0,Z)}$ and $\sigma_{L_u}^{(1,Z)}$, let $\gamma'_u := \sqrt{F(\sigma_{L_u}^{(0,Z)},\sigma_{L_u}^{(1,Z)})}$. By Uhlmann's theorem~\cite{Uhlmann1976TransitionProbability}, there exist purifications $\{\ket*{e'_{b,Z}}_{L_uP_u}\}_b$ of $\{\sigma_{L_u}^{(b,Z)}\}_b$ such that $\vert\braket*{e'_{0,Z}}{e'_{1,Z}}_{L_uP_u}\vert = \gamma'_u$ and, by exploiting the freedom to choose their global phases, we can take $\braket*{e'_{0,Z}}{e'_{1,Z}}_{L_uP_u} = \gamma'_u$ to be real and non-negative. Moreover, since $\gamma'_u \geq \gamma$, we can introduce an additional fictitious ancilla $F_u$ with normalized states $\ket{f_0}_{F_u}$ and $\ket{f_1}_{F_u}$ satisfying $\braket{f_0}{f_1}=\gamma/\gamma'_u$, and then define $\ket*{e_{b,Z}}_{L_uP_uF_u}:=\ket*{e'_{b,Z}}_{L_uP_u}\otimes\ket{f_b}_{F_u}$. These states have overlap exactly $\braket{e_{0,Z}}{e_{1,Z}}=\gamma$, while leaving the reduced state on the leakage system $L_u$ unchanged, $\Tr_{P_uF_u}\!\bigl[\dyad*{e_{b,Z}}\bigr]=\Tr_{P_u}\!\bigl[\dyad*{e'_{b,Z}}\bigr]=\sigma_{L_u}^{(b,Z)}$. Since granting Eve the purifying systems $P_uF_u$ can only increase her information, proving security in this scenario suffices. Renaming $L_uP_uF_u\to L_u$, we may therefore restrict our analysis to pure $Z$-basis leakage states $\sigma_{L_u}^{(b,Z)}=\dyad*{e_{b,Z}}_{L_u}$ with $\braket{e_{0,Z}}{e_{1,Z}}=\gamma$. Importantly, note that this is not a relaxation of the original condition in \cref{eq:leakage-assumption}, since the constructed states $\ket*{e_{b,Z}}_{L_u}$ themselves satisfy \cref{eq:leakage-assumption} with equality.

We remark that the reduction to an overlap of \emph{exactly} $\gamma$---rather than merely \textit{at least} $\gamma$---is what renders the effective channel $\mathcal{E}_\gamma$, and hence the effective phase-error operator in \cref{eq:phase-error-operator-relation}, independent of the round index.

\section{Concentration bounds for sums of dependent random variables}
\label{app:azuma}
\begin{lemma}[Azuma-Hoeffding inequality~\cite{azuma}]\label{lemma:azuma}
    Let $(\rv{M}_u)_{u \geq 0}$ be a martingale with respect to a filtration $(\mathcal{F}_u)_{u \geq 0}$, i.e.,
    \[
    \mathbb{E}[\rv{M}_{u+1} \mid \mathcal{F}_u] = \rv{M}_u \quad \text{for all } u \geq 0,
    \]
    and suppose that the martingale differences are bounded, that is, there exist constants $c_u$ such that
    \[
    |\rv{M}_u - \rv{M}_{u-1}| \leq c_u \quad \text{almost surely for all } u.
    \]
    Then, for any positive integer $N$, and $\epsilon > 0$,
    \[
    \Pr[\rv{M}_N - \rv{M}_0 \geq \Delta_{\rm A}^c(N,\epsilon)] \leq \epsilon,
    \]
    \[
    \Pr[\rv{M}_N - \rv{M}_0 \leq -\Delta_{\rm A}^c(N,\epsilon)] \leq \epsilon,
    \]
    where
    \[
    \Delta_{\rm A}^c(N,\epsilon) = \sqrt{2\sum_{u=1}^N c_u^2 \ln\left(1/\epsilon\right)}.
    \]
\end{lemma}
\begin{corollary}\label{corollary:azuma}
   Let $\rv{X}_1,\dots,\rv{X}_N$ be a sequence of random variables and $(\mathcal{F}_u)_{u \geq 0}$ a filtration. Define $\rv{S}_k:=\sum_{u=1}^k \rv{X}_{u}$ and $\rv{E}_{k}^{\mathcal{F}}:=\sum_{u=1}^{k}\mathbb{E}[\rv{X}_{u}|\mathcal{F}_{u-1}]$. If $0\leq \rv{X}_u \leq 1$ for all $u$, then
\begin{equation}\label{eq:azumaCor}
\begin{split}
    \Pr[\rv{S}_N \geq \rv{E}_{N}^{\mathcal{F}} + \Delta_{\rm A}(N,\epsilon)] & \leq \epsilon,
    \\
    \Pr[\rv{S}_N \leq \rv{E}_{N}^{\mathcal{F}} - \Delta_{\rm A}(N,\epsilon)] & \leq \epsilon,
\end{split}
\end{equation}
where $\Delta_{\rm A}(N,\epsilon) = \sqrt{2N \ln\left(1/\epsilon\right)}$.
\end{corollary}
\begin{proof}
The random variables $\rv{M}_k:=\rv{S}_k-\rv{E}_{k}^{\mathcal{F}}$ form a martingale with respect to $(\mathcal{F}_u)_{u \geq 0}$; applying \cref{lemma:azuma} then yields~\cref{eq:azumaCor}.
\end{proof}
%

\section{Leaky four-state-Bob receivers with detection-efficiency mismatch}
\label{app:four-state}
Receivers in which Bob randomly interchanges the roles of his two detectors in each round have been proposed as a countermeasure against detector-side imperfections, most notably the detection-efficiency mismatch (DEM) problem~\cite{makarov2006effects,qi2005time,fung2009dem,maroySecurityQuantum2010}. In these so-called ``four-state Bob'' schemes, a click in detector ${\rm D}_m$, with $m\in\set{0,1}$, is recorded as the bit $b=m\oplus r$, where $r\in\set{0,1}$ is a random bit that controls a rotation of the input state that effectively interchanges the two detectors. The Achilles' heel of such schemes is that $r$ is encoded using an active modulator, which may suffer from information leakage---e.g., via a THA~\cite{jain2014trojan}. Here, as a particular example, we extend the techniques introduced in the main text to analyze the security of receivers of this type in the presence of a THA that targets the modulators used to encode the random bit $r$ and the basis $\beta$. For simplicity, we restrict ourselves to signals containing at most one photon, as in many previous analyses~\cite{fung2009dem,maroySecurityQuantum2010}, and neglect dark counts. We note, however, that these restrictions could be removed e.g. by combining our technique with the recent methods introduced in~\cite{tupkaryPhaseError2025}.

We consider that the single photon entering Bob's receiver carries the qubit degree of freedom $\Bbob_u$ that encodes the bit---e.g., polarization---together with all remaining degrees of freedom---arrival time, frequency, spatial mode---collected in an auxiliary system $D_u$ with basis states $\ket{d}_D$. That is, the single-photon subspace factorizes as $\mathcal{H}_{\Bbob}\otimes\mathcal{H}_{D}$, and Eve may prepare any state within this subspace, including states entangling the qubit with the modes. Importantly, we allow each detector ${\rm D}_m$ to have an arbitrary unknown detection efficiency profile of the form $\eta_m(d)\in[0,1]$ and we encode each profile in a positive operator on $D_u$, $E^{m}_{D}:=\sum_{d}\eta_{m}(d)\,\dyad{d}_{D}$, with $0\leq E^{m}_{D}\leq\identity$. The POVM element for a detection with recorded bit $b\in\set{0,1}$, conditioned on $(\beta,r)$, is then $F^{b|\beta,r}_{\Bbob D}=\dyad{b_{\beta}}_{\Bbob}\otimes E^{b\oplus r}_{D}$, and the three-outcome POVM is completed by the no-detection element $\identity_{\Bbob D}-\sum_{b}F^{b|\beta,r}_{\Bbob D}$.

Similarly to the main text, we consider that, after his measurement, Bob's receiver emits a system $L_u'$ in a state $\sigma_{L'}^{(\beta,r)}$ that may depend on both settings, and we assume only that, for all rounds $u$,
\begin{equation}
\label{eq:fourstate-assumption}
  F\big(\sigma_{L'}^{(Z,0)},\sigma_{L'}^{(Z,1)}\big)\geq\gamma^2,
\end{equation}
for some $\gamma>0$. Pessimistically handing Eve a purifying system of each emitted state and applying the reduction of \cref{app:purification}, we take the leakage states to be pure, i.e., $\sigma_{L'}^{(\beta,r)}=\dyad*{w^{\beta}_{r}}_{L'}$, with $\braket*{w^{Z}_{1}}{w^{Z}_{0}}=\gamma\geq 0$. Similarly to the main text, the leakage emission is modeled with an isometry
\begin{equation}\label{eq:leakage_isometry_4stateBob}
    \hat{W}^{\beta}\ket{r}_{R}=\ket{r}_{R}\ket*{w^{\beta}_{r}}_{L'}
\end{equation}
that this time is applied to a quantum coin $R_u$ that Bob initializes in the state $\ket{+}_{R}$ and uses to decide $r$ by performing a $Z$-basis measurement. This isometry commutes with Bob's coin measurement, so we can consider that $\hat{W}^{\beta}$ is applied upon reception.

Now, we split Bob's $\beta$-basis measurement into its announced part---the detection flag, which must remain identical in the virtual protocol---and its private part---which determines Bob's final outcomes. In particular, we define the detection operator
\begin{equation}
\label{eq:fourstate-filter}
\begin{split}
  N^{\beta}_{\Bbob DR}&:=\sum_{b,r}\dyad{b_{\beta}}_{\Bbob}\otimes E^{b\oplus r}_{D}\otimes\dyad{r}_{R}
  \\
  &=\bar{E}_{D}\otimes\identity_{\Bbob R}+\frac{1}{2}\,\Delta E_{D}\otimes Z^{\beta}_{\Bbob}\otimes Z_{R},
\end{split}
\end{equation}
where $Z^{\beta}_{\Bbob}:=\dyad{0_{\beta}}-\dyad{1_{\beta}}$, $Z_R:=\dyad{0}-\dyad{1}$, and we have introduced the mean and mismatch operators 
\begin{equation}
\bar{E}_{D}:=\tfrac{1}{2}\big(E^{0}_{D}+E^{1}_{D}\big)
\quad\text{ and }\quad
\Delta E_{D}:=E^{0}_{D}-E^{1}_{D},
\end{equation}
respectively. Since $N^{\beta}_{\Bbob DR}$ is diagonal in the basis $\set{\ket{b_\beta}\ket{d}\ket{r}}_{b,d,r}$, Bob's actual detector-and-coin measurement in a $\beta$-basis round decomposes as the two-outcome filter $\set{N^{\beta}_{\Bbob DR},\identity-N^{\beta}_{\Bbob DR}}$, with Kraus operator $\sqrt{N^{\beta}_{\Bbob DR}}$, yielding the announced flag, followed by the projective readout of $(Z^{\beta}_{\Bbob},Z_R)$, which produces the private outcomes $(b,r)$. 

In the virtual protocol, Bob can replace, in the detected key rounds, the final projective readout by any virtual measurement to compute the phase error rate. In particular, here we take the virtual observable $X_{\Bbob}\otimes X_R$, which commutes with the detection operator $N^{Z}_{\Bbob DR}$. Phase errors are therefore defined as disagreements between Alice's virtual $X$-basis outcome and Bob's $X_{\Bbob}\otimes X_{R}$ outcome. In particular, we define the operators
\begin{equation}
\label{eq:fourstate-operators}
\begin{gathered}
  \hat{O}^{\rm ph}_{A\Bbob DR}:=\sqrt{N^{Z}_{\Bbob DR}}\,\hat{E}^{\rm ph}_{A\Bbob DR}\sqrt{N^{Z}_{\Bbob DR}}=\hat{E}^{\rm ph}_{A\Bbob DR}\,N^{Z}_{\Bbob DR},
  \\[2pt]
  \hat{E}^{\rm ph}_{A\Bbob DR}:=\tfrac{1}{2}\big(\identity-X_A\otimes X_{\Bbob}\otimes X_R\big)\otimes\identity_D,
  \\[6pt]
  \hat{O}^{\rm x}_{A\Bbob DR}:=\sqrt{N^{X}_{\Bbob DR}}\,\hat{E}^{\rm x}_{A\Bbob DR}\sqrt{N^{X}_{\Bbob DR}}=\hat{E}^{\rm x}_{A\Bbob DR}\,N^{X}_{\Bbob DR},
  \\[2pt]
  \hat{E}^{\rm x}_{A\Bbob DR}:=\tfrac{1}{2}\big(\identity-X_A\otimes X_{\Bbob}\big)\otimes\identity_{DR},
\end{gathered}
\end{equation}
where we use $\big[N^{Z}_{\Bbob DR},\,X_{\Bbob}\otimes X_{R}\big]=0$ and $[X_\Bbob,Z^X_\Bbob]=0$. Then, conditioned on the record $\mathcal{F}_{u-1}$, with $\rho\equiv\rho^{\mathcal{F}_{u-1}}_{A_u\Bbob_uD_uR_u}$ being the round-$u$ conditional state, the conditional phase-error and $X$-basis-error probabilities are given by 
\begin{equation}
\begin{split}
    \Pr[\rv{\chi}_{\rm ph}^{u}|\mathcal{F}_{u-1}]&=\piA{Z}\pB{Z}\Tr[\hat{O}^{\rm ph}_{A\Bbob DR}\,\rho],
    \\
    \Pr[\rv{\chi}_{\rm x,er}^{u}|\mathcal{F}_{u-1}]&=\piA{X}\pB{X}\Tr[\hat{O}^{\rm x}_{A\Bbob DR}\,\rho].
\end{split}
\end{equation}

Now, note that Eve's round-$u$ decision on the state that is forwarded to Bob is independent of his settings $(\beta,r)$.
Thus, after tracing out $L_u'$ and Eve's system $E_u$, the state of the remaining systems in a $\beta$-basis round has the product form $\rho=\sigma_{A\Bbob D}\otimes\rho^{\beta}_{R}$, where $\sigma_{A\Bbob D}$ is arbitrary but common to the two basis choices, and $\rho^{\beta}_{R}:=\Tr_{L'}\big[\hat{W}^{\beta}\dyad{+}_{R}\hat{W}^{\beta\,\dagger}\big]$. In particular, we have that $\rho^{Z}_{R}=\tfrac{1}{2}\big(\identity_{R}+\gamma\,X_{R}\big)$. Expanding \cref{eq:fourstate-operators} with \cref{eq:fourstate-filter} and using $\Tr[Z_{R}\,\rho^{\beta}_{R}]=0$, $\Tr[X_{R}Z_{R}\,\rho^{Z}_{R}]=0$ and $\Tr[X_{R}\,\rho^{Z}_{R}]=\gamma$, every term containing $\Delta E_{D}$ vanishes, and one obtains
\begin{equation}
\label{eq:fourstate-traces}
\begin{gathered}
  \Tr\big[\hat{O}^{\rm ph}_{A\Bbob DR}\,\rho\big]
  =\tfrac{1}{2}\Tr\big[\big(\identity_{A\Bbob}-\gamma\,X_A\otimes X_{\Bbob}\big)\otimes\bar{E}_{D}\,\sigma\big],
  \\[2pt]
  \Tr\big[\hat{O}^{\rm x}_{A\Bbob DR}\,\rho\big]
  =\tfrac{1}{2}\Tr\big[\big(\identity_{A\Bbob}-X_A\otimes X_{\Bbob}\big)\otimes\bar{E}_{D}\,\sigma\big],
\end{gathered}
\end{equation}
and
\begin{equation}\label{eq:fourstate-traces2}
    \Tr\big[N^{Z}_{\Bbob DR}\,\rho\big]=\Tr\big[N^{X}_{\Bbob DR}\,\rho\big]=\Tr\big[\big(\identity_{A\Bbob}\otimes\bar{E}_{D}\big)\,\sigma\big]\green{.}
\end{equation}
Hence 
\begin{equation}\label{eq:fourstate-operator-difference}
\begin{split}
    \Tr[(\hat{O}^{\rm ph}-\hat{O}^{\rm x})\rho]&=\tfrac{1-\gamma}{2}\Tr[(X_A\otimes X_{\Bbob}\otimes\bar{E}_{D})\,\sigma]
    \\
    &
    \leq
    \tfrac{1-\gamma}{2}\Tr[(\identity_{A\Bbob}\otimes\bar{E}_{D})\,\sigma],
\end{split}
\end{equation}
where we have used $\bar{E}_{D}\geq0$ and $\norm{X_A\otimes X_{\Bbob}}_{\infty}=1$. Moreover, since the detection probability is basis independent (see \cref{eq:fourstate-traces2}) we can divide by $\Tr[(\identity_{A\Bbob}\otimes\bar{E}_{D})\,\sigma]$ in both sides of~\cref{eq:fourstate-operator-difference} and multiply by the basis-choice probabilities, leading, for every detected round $u$, to
\begin{equation}
\label{eq:fourstate-final}
\begin{split}
  \Pr[\rv{\chi}_{\rm ph}^{u}|\mathcal{F}_{u-1},\rv{D}_u]
  \leq\,
  &
  \frac{\piA{Z}\pB{Z}}{\piA{X}\pB{X}}\,\Pr[\rv{\chi}_{\rm x,er}^{u}|\mathcal{F}_{u-1},\rv{D}_u]
  \\
  &
  +\piA{Z}\pB{Z}\,\frac{1-\gamma}{2},
\end{split}
\end{equation}
where $\rv{D}_u$ represents the detection event. Since the coefficients in \cref{eq:fourstate-final} are state-independent, indexing the random variables by the \emph{detected} rounds and applying Azuma-type inequalities~\cite{azuma,kato,mannalath2026quantum}, as in the main text, bounds the number of phase errors by the rescaled observed number of $X$-basis errors plus $\tfrac{1-\gamma}{2}$ times the number of detected \emph{key} rounds, with deviation terms scaling with the number of detected rounds.

\bibliography{refs}

\end{document}